\documentclass{elsart}

\usepackage{amsfonts}
\usepackage{amsmath}
\usepackage{amssymb}
\usepackage[all]{xy}
\usepackage{graphicx}
\usepackage{amscd}

\usepackage{color}

\input xy
\xyoption{all}

\newtheorem{theorem}{Theorem}[section]
\newtheorem{corollary}[theorem]{Corollary}

\newtheorem{proposition}[theorem]{Proposition}

\theoremstyle{definition}

\theoremstyle{remark}

\newtheorem{example}[theorem]{Example}

\newenvironment{proof}[1][Proof.]{\begin{trivlist}
\item[\hskip \labelsep {\it #1}]}{$\hfill\square$ \end{trivlist}}

\newcommand{\ds}{\displaystyle}
\newcommand{\nlb}{\nolinebreak}

\newcommand{\F}{\mathbb{F}}

\newcommand{\N}{\mathbb{N}}

\newcommand{\C}{\mathcal{C}}

\newcommand{\df}{d_{free}}

\begin{document}

\begin{frontmatter}

\title{Receding horizon decoding of convolutional codes\thanksref{labeltitle}}
\author[label1]{Jos\'e Ignacio Iglesias Curto\corauthref{cor1}}
\author[label2]{Uwe Helmke}

\thanks[labeltitle]{Research partially supported by DFG-SPP 1305 grant HE
1858/12-1 and by Junta de Castilla y Le\'on research project
SA029A08.
\\Email: joseig@usal.es, helmke@mathematik.uni-wuerzburg.de}
\corauth[cor1]{Corresponding author.}

 \address[label1]{Dept. de Matem\'aticas, University of Salamanca, 37008 Salamanca, Spain}
 \address[label2]{Institute of Mathematics, University of W{\"u}rzburg, 97074 W{\"u}rzburg, Germany}

\begin{abstract}
Decoding of convolutional codes poses a significant challenge for
coding theory. Classical methods, based on e.g. Viterbi decoding,
suffer from being computationally expensive and are restricted
therefore to codes of small complexity. Based on analogies with
model predictive optimal control, we propose a new iterative method
for convolutional decoding that is cheaper to implement than
established algorithms, while still offering significant error
correction capabilities. The algorithm is particularly well-suited
for decoding special types of convolutional codes, such as e.g.
doubly cyclic convolutional codes.
\end{abstract}

\end{frontmatter}

\section{Introduction}
A central aim of coding theory is to protect transmitted or stored
information against errors. A common technique is to split the
sequence of information symbols into blocks of constant length and
map each block injectively to a codeword of larger length. This map
is called the encoding map and its image space is referred to as a
block code. To protect the process from transmission errors requires
a procedure that enables one to recover the sent message from the
received one by projecting it back to the code. This is called
decoding and forms the basis of error correction algorithms of
current codes.  Decoding is an inherently difficult task, but for
block codes effective decoding algorithms are available that depend
on the special algebraic structure of the codes; e.g. BCH codes,
Reed-Solomon codes, or list decoding techniques, see e.g.
\cite{MS:77, LC:83}.

In the sequel, we will consider block codes only as an intermediate
step for decoding convolutional codes. Such \emph{convolutional
codes} are a natural generalization of block codes and have found
widespread applications; see \cite{McE:98,Pir:88}. Algebraically,
they are defined as submodules of $\F[z]^n$, spanned by the columns
of a full rank rectangular polynomial matrix $G(z)\in \F^{n\times
k}[z]$. In contrast to block codes, where a rather rich theory is
available, there is no really efficient decoding algorithm known for
convolutional codes. Classical decoding algorithms for convolutional
codes such as e.g. Viterbi decoding~\cite{LC:83} work well only for
codes of moderate dimensions. There is thus considerable interest in
developing efficient decoding algorithms for convolutional codes
that improve known algorithms, such as Viterbi decoding.

Since  convolutional codes can be interpreted as linear control
systems defined over a finite field $\F$, see the survey paper
\cite{Ros:01} and the references therein, it is possible to apply
the rich tools from linear systems theory and optimal control to
study such codes. This approach offers a better understanding of
convolutional codes and led already to new construction methods and
algorithms; see \cite{MS:77, Glu:05,HRS:05,Ros:01,RSY:96}. From a
systems theory point of view, decoding of convolutional codes can be
interpreted in at least two different ways \cite{Ros:97}. One
interpretation treats decoding as a tracking problem, where the
decoder attempts to track the received message by the most probable
codeword sent. Another perspective lies in treating it as a
filtering problem, where the decoder is requested to filter the
noise introduced by the channel.

In this work we will follow the first approach, by treating the
decoding problem of convolutional codes as a tracking exercise for
linear systems to replace the received message by its closest
polynomial codeword \cite{Ros:97}. Motivated by analogies with model
predictive control, we propose a simple receding horizon algorithm
for convolutional decoding, that combines an arbitrary block
decoding algorithm with a few receding horizon steps. The proposed
method has the advantage of being computationally cheaper than
well-known Viterbi decoding, while still achieving substantial error
correction. In a companion paper \cite{GHC:09} we show that our
algorithm can lead to competitive decoding results for the class of
doubly-cyclic convolutional codes.

The paper is structured as follows. Section 2 introduces basic
terminology from coding theory and outlines a decoding procedure for
convolutional codes via the Bellman optimality principle.
In Section 3, the new algorithm is proposed and error correction
capabilities of the method are established. Examples are discussed in
section 4 to illustrate the decoding properties of the algorithm.

\section{Dynamic programming approach to convolutional decoding}

\subsection{Coding theory terminology}
We begin with a brief summary of basic notions from coding theory.
For further information we refer to standard textbooks as, e.g.,
\cite{MS:77,McE:77}. Throughout this paper, $\F$ denotes a finite
field with $q$ elements. Given any full column rank matrix $G\in
\F^{n\times k}$, a \emph{linear block code} of length $n$ and
dimension $k$ is a $k$-dimensional vector space
\begin{equation*}
\C:=\{Gv| v\in \F^k\}\subset\F^n
\end{equation*}
and $G$ is called a \emph{generator matrix} of $\C$. If, possibly
after a permutation of the rows, $G\in \F^{n\times k}$ is of the
form
\begin{equation*}
G=\begin{pmatrix}
    A \\
    I_k
\end{pmatrix},
\end{equation*}
then $G$ is called a \emph{systematic generator matrix}. The
\emph{check matrix} of the code then is defined as
\begin{equation*}
S=\begin{pmatrix}
    I_{n-k}& -A
\end{pmatrix}
\end{equation*}
and, since $\C=\rm{Ker} S$, it provides a kernel description of the
code. A natural metric on a code is defined by the Hamming distance.
The \emph{Hamming distance}  of vectors $x,y\in\F^n$  is defined by
$d(x,y):= w(x-y)$, where
\begin{equation*}
    w(c):=\#\{i|c_i\neq 0\}.
\end{equation*}
denotes the \emph{Hamming weight}.
 The \emph{minimum distance} of a code $\C$ then is the minimum among the distances between any pair
of codewords, i.e.
\begin{equation*}
    \rm{d}(\C):=\min\limits_{c\in \C\setminus\{0\}}w(c).
\end{equation*}
It measures the error correcting capabilities for the code. In case
a received vector $x\in \F^n$ does not belong to $\C$, i.e. if some
error has occurred in the encoding process, a natural way to recover
the sent message is by taking the maximum likelihood estimate of
$x$, i.e. to replace $x$ by the closest vector from $\C$. This
process is called \emph{decoding}. A decoder is a map
$\pi:\F^n\rightarrow \C$ that is the identity on $\C$. A special
case in point are maximum likelihood decoders that are defined by
the set valued map
\begin{equation*}
    \rm{\pi}_{\C}(x):=\arg\min\limits_{c\in
    \C}d(x,c)\subset \C.
\end{equation*}

In order to achieve the best error correction properties of
decoding, the vectors from $\C$ should be as far apart as possible.
Thus one is searching for codes whose minimal distance is a large as
possible.  The maximum number of errors that a code can correct is
known as the \emph{error correction capacity} of the code, which is
half the minimum distance of the code minus one. If a code of length
$n$ has minimum distance $d= \rm{d}(\C)$, then the spheres centered
at the codewords with radius $\lfloor \frac{d-1}{2}\rfloor$ are
disjoint. Any vector from $\F^n$ contained in one of these spheres
can thus be decoded to its unique nearest codeword, the center of
the sphere. But the union of these spheres may not cover all vectors
of the ambient space. The \emph{covering radius} $\rho_\C$ of a code
$\C$ is the maximum distance from any vector of $\F^n$ to its
nearest codeword. The covering radius of a code is the smallest
radius needed for spheres centered at the codewords to cover the
whole ambient space.

\subsection{Convolutional codes and Bellman principle}

In this paper we follow the well-established approach, see e.g.
\cite{Ros:97}, that regards a convolutional code $\C$ of length $n$,
dimension $k$ and complexity $\delta$ as a submodule
\begin{equation*}
 \C:=\{G(z)v(z)| v(z)\in \F^k[z]\}\subset \F^n[z],
 \end{equation*}
defined by a full column rank polynomial matrix $G(z)\in \F^{n\times
k}[z]$, $k\leq n$. Here the complexity $\delta$ of $\C$, also referred
to as the ``constraint length'' of the code~\cite{O:68}, is defined
as the maximal degree of all $k\times k$ minors of $G(z)$.

After a suitable permutation of the rows, we can assume that the
generator matrix is of the form
\begin{equation}
\label{eq:gen}
     G(z)=\left(
   \begin{array}{c}
     P(z)_{}\\
     Q(z)_{}
   \end{array}
   \right)
 \end{equation}
with right coprime
   polynomial factors $P(z)\in \F^{(n-k)\times k}$ and $Q(z)\in
   \F^{k\times k}$, respectively. Here $\delta=\rm{deg} det Q(z)$ is
   assumed to be the maximal degree of all $k\times k$ minors of $G$.
Therefore the transfer function $P(z)Q(z)^{-1}\in\F^{(n-k)\times
k}(z)$ is proper rational of McMillan degree $\delta$ and thus has a
minimal (i.e. controllable and observable) state space realization
\begin{equation}\label{eq:systemD}
  \begin{array}{l}
    x_{t+1}=A x_t + B u_t, \quad  x_0=0\\
    y_t = C x_t + D u_t,
  \end{array}
\end{equation}
$A\in \F^{\delta \times \delta}$, $B\in\F^{\delta \times k}$,
$C\in\F^{(n-k)\times \delta}$, $D\in\F^{(n-k)\times k}$,
  $x\in\F^{\delta},u\in\F^{k} y\in\F^{n-k}$.
Conversely, given any such linear linear systems representation over
the field $\F$, then a right coprime factorization of the transfer
function $C(zI_{\delta}-A)^{-1}B+ D=P(z)Q(z)^{-1}$ defines a
convolutional code with generator matrix (\ref{eq:gen}). Thus the
complexity of the code corresponds to the McMillan degree of the
associated rational transfer function, i.e. to the dimension of the
state space of the associated controllable and observable linear
system.

In the above framework, codewords of a convolutional code correspond
to polynomials $c(z)=\sum c_i z^i\in\C$ with a finite number of
coefficients $c_0,\cdots, c_{\gamma}\in\F^n$. We extend the Hamming
distance on $\F^n$ to a metric on vector polynomials $\F^n[z]$ via
$\rm{dist}(\tilde{c}(z),c(z)):= \sum_{i=0}^{\infty}d(\tilde{c}_i,c_i)$.
Given any polynomial  $\tilde{c}(z)=\sum_{i=0}^{T}\tilde{c}_i z^i\in\F^n[z]$
the task of minimal distance decoding then asks to find a code word
$c_{\rm{opt}}\in\C$ that minimizes the distance to $\tilde{c}$, i.e.
\begin{equation*}
c_{\rm{opt}}:=\arg\min\limits_{c\in\C}\rm{dist}(\tilde{c},c)
\end{equation*}

Such an optimal codeword always exists, since $\F^n$ is finite, but
need not be unique. The map $\tilde{c}\mapsto c_{\rm{opt}}$ is
called a minimal distance decoder for $\C$. The method of dynamic
programming yields a way to calculate this optimal code word
$c_{\rm{opt}}$ for each specific choice of $\tilde{c}$.

To see how this works we first give a linear systems interpretation of
the code words of $\C$. This requires that (\ref{eq:systemD}) is
observable.
Rosenthal \cite{Ros:97} has shown that a
polynomial $c(z)=\sum_{i=0}^{\gamma} c_i z^i$ is a code word of $\C$
if and only if and only if
\begin{equation}
\label{eq:rosen}
  \begin{matrix}
    \begin{pmatrix}
      0 & A^{\gamma}B & A^{\gamma-1}B & \ldots & \ldots B\\
        &  D  \\
        &  CB & D & \\
    -I  &  CAB & CB & D\\
        & \vdots &  \ddots & \ddots &  \\
        & CA^{\gamma-1}B& \ldots & \ldots & CB & D \\
    \end{pmatrix}
\begin{pmatrix}
y_0\\
\vdots\\
y_{\gamma}\\
u_0\\
\vdots\\
u_{\gamma}\\
\end{pmatrix}=0,
  \end{matrix}
\end{equation}

where $I$ denotes the $(\gamma+1)(n-k)\times (\gamma+1)(n-k)$-
identity matrix and
\begin{equation*}
c_t=\begin{pmatrix}
    y_t\\ u_t
 \end{pmatrix}\in\F^{(n-k)}\times\F^{k}, \quad t=0,\cdots, \gamma.
\end{equation*}

The input sequences $u_0,\cdots,u_{\gamma}$ in (\ref{eq:systemD})
that define an admissible  code word $c$ are therefore just those
controls that steer the initial condition $x_0=0$ in finite time
back to $x_0=0$. It is easily seen that there exist always a
nontrivial choice of such input sequences.  Moreover, the minimal
time to steer back to $x_0=0$ via a non-zero input is at least
$\gamma\geq \kappa_{\rm{min}}$; $\kappa_{\rm{min}}\leq
\lfloor\frac{\delta}{k} \rfloor $ being the smallest controllability
index of (\ref{eq:systemD}). Let ${\cal U}=\rm{proj}(\C)$ denote the
submodule of $\F^k[z]$, obtained by projection of $\C$ onto
$\F^k[z]$. By coprimeness of $P,Q$ it follows that
\begin{equation*}
{\cal U}:= Q(z)\F^k[z]=\{Q(z)u(z)\in \F^k[z]\;| u(z)\in \F^k[z]\; \}.
\end{equation*}
Moreover, expressed in terms of controllable and observable
realizations $(A,B,C,D)$ of $P(z)Q(z)^{-1}$, we obtain the state
space description
\begin{equation*}
{\cal U}= \{u(z)=\sum_{i=0}^{\gamma} u_i z^i\in \F^k[z]\;|
\sum_{i=0}^{\gamma}A^iBu_i=0, \;\gamma \in \N_0 \; \}.
\end{equation*}
In linear systems theory, ${\cal U}$ is therefore referred to as the
module of zero return.

Now assume, that we want to decode a polynomial $\tilde{c}\in
\F^n[z]$ via (\ref{eq:systemD}) with coefficients
\begin{equation*}
\tilde{c_t}=\begin{pmatrix}
    \tilde{y}_t\\ \tilde{u}_t
 \end{pmatrix}\in\F^{(n-k)}\times\F^{k}, \quad t=0,\cdots, T
\end{equation*}
By the above, this means to solve the optimal control problem of
finding an admissible input function $u\in {\cal
  U}$ that minimizes the tracking error cost functional
\begin{equation*}
\begin{array}{rl}
    J(x_0,u):=\rm{dist}(\tilde{c}, c)&= \ds{\sum_{t=0}^{\infty}} k_t(x_t,u_t)
\end{array}
\end{equation*}
where
\begin{equation}
\label{eq:ka}
k_t(x,u):= w(u-\tilde{u}_t)+w(Cx-\tilde{y}_{t}+Du), \quad t\in\N_0.
\end{equation}
Note that the above series is always finite, as the inputs $u(z)$
are constrained to be admissible polynomials. Thus this is an
optimal control problem where the classical $l^2$-distance from
linear quadratic controller design is replaced by the Hamming
distance. Let
\begin{equation}
V_{\infty}(x_0)=\inf_{u\in {\cal U}}  J(x_0,u)
\end{equation}
denote the optimal value function. The optimal control can then be
computed via the Bellman principle, although this is a bit complicated
here due to the varying length of the inputs. Thus we do not pursue
this approach here. A simplified analysis
can be given under the assumption that all data are available over a
fixed horizon $T$. Thus assume, we consider the task of minimizing
\begin{equation*}
 J(x_0,u({\cdot}), T)=\sum_{t=0}^{T}k_t(x_t,u_t)
\end{equation*}
with $k_t(x,u)$ as in (\ref{eq:ka}) and we optimize over all input
sequences $u_0,\cdots, u_T$. For $N=0,\cdots,T$ let
\begin{equation*}
V_N(x)=\min \sum_{t=N}^{T}k_t(x_t,u_t)
\end{equation*}
denote the $N$-th value function, $x_N:=x$, where minimization
occurs over all sequences $u_N,\cdots,u_T$.  By the Bellman
principle, these functions satisfy the functional equation
\begin{equation*}
V_N(x)=\min_{u} \{k_{N}(x,u)+V_{N+1}(Ax+Bu)\}; \quad N=0,\cdots,T-1.
\end{equation*}
For
\begin{equation*}
u_{N}(x)=\arg\min_{u} \{k_{N}(x,u)+V_{N+1}(Ax+Bu)\}
\end{equation*}
the optimal control strategy for $x_0=0$ then becomes
\begin{equation*}
x_{t+1}=Ax_t+Bu_{t}(x_t), \quad t=0,\cdots, T-1.
\end{equation*}
Applied this approach to the negative log-likelihood function of the
channel output, this then becomes exactly the Viterbi decoding
algorithm~\cite{O:68}. However, the computation of the value
functions is still computationally expensive and success of this
method is therefore restricted to codes of small complexity. In the
next section we will give a somewhat easier model predictive control
approach.

\section{Convolutional decoding via receding horizon}

We now show how a modification of the \emph{receding horizon method}
from optimal control leads to an effective decoding algorithm for
convolutional codes. The solution is achieved by means of a slight
modification of the classical receding horizon method were $L$,
instead of 1, vectors are taken to update the solution. Moreover,
we carry
out the main minimization step by a suitable decoding algorithm for
a certain block code associated to the system. Throughout this
section we assume that (\ref{eq:systemD}) is controllable and
observable. This is actually a natural and frequently used
assumption \cite{Ros:97}.

We assume the following data are given:
\begin{enumerate}
\item  A finite sequence of received, to be decoded, words
 \begin{equation*}
 \tilde{c}_t:=\begin{pmatrix}
     \tilde{y}_t\\ \tilde{u}_t
  \end{pmatrix}\in\F^{(n-k)}\times\F^{k}, \quad t=0,\cdots, T
 \end{equation*}
 for arbitrary $T\in\N_0$.
\item Positive integers $L\leq N \leq T$.

\end{enumerate}
We then attempt to iteratively minimize the finite cost function
\begin{equation}\label{eq:finitecost}
    J(x_0, u, T)=\ds{\sum_{i=0}^{T-1}}
    (w(y_{i}-\tilde{y}_{i})+w(u_{i}-\tilde{u}_{i}))
\end{equation}
over the set of admissible inputs $u\in {\cal U}$. For this we note,
by controllability of (\ref{eq:systemD}), that any unconstrained
minimum $u^*$ of the
cost function (\ref{eq:finitecost}) can be extended to an admissible
input $u^*_{\rm{adm}}\in{\cal U} $ by steering the terminal state
$x_T$ to zero. Of course, this admissible solution  $u^*_{\rm{adm}}$
will not necessarily be an optimal solution of (\ref{eq:finitecost}).

The iteration steps to be followed at every time instant $t\in \N_0$ are
\begin{enumerate}
\item
Consider the initial (known) state as $x_t$.
\item
Solve an $N$-step finite horizon tracking problem, i.e., find the
unconstrained input sequence $\{u_{t+i}\}_{i=0}^{N-1}$ which minimizes
\begin{equation*}
    J(x_t, u, N)=\ds{\sum_{i=0}^{N-1}}
    [w(y_{t+i}-\tilde{y}_{t+i})+w(u_{t+i}-\tilde{u}_{t+i})]
\end{equation*}
\item
Update the solution input with $\{u_t,\ldots,u_{t+L-1}\}$ and use it
to update the solution output with $\{y_t,\ldots,y_{t+L-1}\}$ and to
calculate $x_{t+L}$.
\item
Update the time instant $t$ with $t+L$ until $t=T$.
\end{enumerate}

The last step of the algorithm then results in an input sequence
$u_0,\cdots, u_{T-1}$. By controllability of $(A,B)$, we can extend this
sequence to an admissible input sequence  $u_0,\cdots, u_T, u_{T+1},
\cdots, u_{T+\tau}$, by steering the final state $x_T$  into $x_{T+\tau}=0$.
Here, $\tau\leq \kappa_{\rm{max}}$, with $\kappa_{\rm{max}}$ the
largest controllability index of $(A,B)$.

The obtained input sequence  $u_0,\cdots, u_{T+\tau}$ with
associated outputs $y_0,\cdots, y_{T+\tau}$ and
\begin{equation*}
 c_t:=\begin{pmatrix}
     y_t\\ u_t
  \end{pmatrix}\in\F^{(n-k)}\times\F^{k}, \quad t=0,\cdots, T+\tau,
 \end{equation*}
then defines a codeword $c(z)=\sum_{i=0}^{T+\tau}c_iz^i\in \C$ that
serves as a decoding for $\tilde{c}$. We emphasize again, that due
to the unconstrained minimization, this is not an optimal solution
to the above tracking problem.

Step 2 represents the main problem to be solved, which we now
replace by a general block decoding step. By inspection,
$J(x_t,c,N)=w(\xi_N)$, for a vector $\xi_N=z_{t,N}+B_Nu_{t,N}\in
\F^{Nn}$ defined by
\begin{equation*}\scriptsize
    z_{t,N}:=
    \begin{pmatrix}
    CA^{N-1}\\\vdots\\C\\0\\ \vdots \\ 0 \\
    \end{pmatrix}
    x_t -
    \begin{pmatrix}
    \tilde{y}_{t+N-1}\\\vdots\\ \tilde{y}_t\\
    \tilde{u}_{t+N-1}\\\vdots\\ \tilde{u}_t
    \end{pmatrix},
  \quad
    B_N=
    \begin{pmatrix}\footnotesize
      D & CB & CAB & \ldots & CA^{N-2}B \\
      0 & D & CB & \ldots & CA^{N-3}B \\
      \vdots & \ddots & \ddots & \ddots & \vdots \\
      \vdots & & \ddots & \ddots & CB \\
      0 & \dots & \dots & 0 & D \\
      & & & &  \\
      \multicolumn{5}{c}{I_{Nk}} \\
      & & & &
    \end{pmatrix}, \quad
 u_{t,N}=
  \begin{pmatrix}
  u_{t+N-1}\\ u_{t+N-2}\\ \vdots\\u_t.
  \end{pmatrix}
\end{equation*}

Thus, each minimization step of length $N$ can be solved by decoding
the vector $z_{t,N}$ with respect to the block code $\C_N$ generated
by $B_N$.
At this point we make contact with coding theory.
Consider the block code
\begin{equation*}
\C_{N}=\{B_Nu \;| \; u\in\F^{Nk}\}\subset \F^{Nn}
\end{equation*}
with generator matrix $B_N$. The maximum likelihood decoder of
$\C_{N}$ then is
\begin{equation*}
\rm{\pi}_{\C_{N}}(z_{t,N}):=\arg\min\limits_{v\in
    \C_{N}}d(z_{t,N},v).
\end{equation*}
In particular, if we decompose any vector
$v\in\rm{\pi}_{\C_{N}}(z_{t,N})$ as $v={v_1\choose v_2}\in
\F^{n-k}\times\F^k$, then the optimal vector $u$ is equal to $-v_2$.

By replacing step (2) by a suitable decoding algorithm for $\C_N$,
we obtain the following algorithm where decoding of $\C_N$ is
represented by the function \texttt{Decoding}. Note that, for ease
of notation, we have written the column vectors $w,z-e,u$ in row
vector form.

\noindent{} \boxed{
\begin{minipage}{\textwidth}\textbf{Receding Horizon Decoding Algorithm}\\
 \texttt{\noindent{}Input:
\begin{itemize}
\item[*]
$x_0$, $\{\tilde{y}_t\}_0$, $A$, $B$, $C$, $D$ (*they describe the
system*),
\item[*]
$N$ (*length of finite horizon*)
\end{itemize}
U = Empty\_List (*the solution sequence*)}\\
$t=0$\\
\texttt{\noindent{}
 Wile Exit=NO}

\hspace{.7cm} $w=(CA^{N-1}x_t-\tilde{y}_{t+N-1},\ldots,C
x_t-\tilde{y}_{t},0,\ldots,0)$

\hspace{.7cm} $e=$\texttt{Decoding}$(B_N, w)$

\hspace{.7cm} $(c'_N,\ldots,c'_1,c''_N,\ldots,c''_1) = z - e$

\hspace{.7cm} $u = -({c''_1,\ldots,c''_L})$

\hspace{.7cm} \texttt{U = Append $u$ to U}

\hspace{.7cm} $x_{t+L}=A^Lx_{t} - \sum_{i=1}^L
A^{L-i}Bc''_i$

\hspace{.7cm} ${t=t+L}$

\hspace{.7cm} \texttt{If ${t>T}$, Exit=YES}\\
\texttt{End-While}\\
\hspace{.7cm} \texttt{Extend U to an admissible input sequence}\\
\texttt{Output: U}
\end{minipage} }

The whole algorithm takes at most $\lceil\frac{T}{L}\rceil$ many
steps. Parameter $L$ is directly related with the number of $u_j$
that can be correctly decoded at each $N$-step problem. The precise
relationship depends on the code $\C_N$ and it is given in Theorem
\ref{th:admisible_errors}. Moreover, in each step the decoding
problem may have more than one solution. We address this issue in
the final section of the paper. Although convergence to the optimal
solution of the tracking problem may not occur, we can at least
derive an upper bound for the achieved cost.

\begin{proposition}
Let $\rho_N$ be the covering radius of the code $\C_{N}$ and let $\tilde{u}$ denote any input
sequence that is produced by the algorithm. Then
\begin{equation*}
     J(x_0, \tilde{u}, T)\leq \left\lceil\frac{T}{L}\right\rceil \rho_N
\end{equation*}
\end{proposition}

\begin{proof}
For any vector $z\in\F^{Nn}$, there exists a codeword at distance less
or equal to $\rho_N$. Thus the cost added to the functional at every
iteration step $t=0,L,2L,\cdots,\lceil\frac{T}{L}\rceil$ is upper
bounded by
 \begin{equation*}
     \underset{u}{min}\{w(z_{t,N}+B_Nu)\}=w(e)\leq \rho_N.
 \end{equation*}
Therefore, the total cost obtained after $\lceil\frac{T}{L}\rceil$ steps is
upper bounded by $\lceil\frac{T}{L}\rceil \rho_N$.
\end{proof}

The number $L$ of steps to update can be chosen in dependence on the
correction properties of the code $\C_{N}$. Recall the notion of
\emph{decoding error}, which occurs when the decoding algorithm
outputs a codeword different to the original one. A case in point
here is where in the iterative method the solution is not updated
with the entire codeword from $\C_{N}$ but with the components
corresponding to $u_t,\ldots,u_{t+L-1}$. Those decoding errors which
do not affect these components will be called admissible decoding
errors. The precise connection between $L$ and $N$ is given by the
following theorem.

\begin{theorem}\label{th:admisible_errors}
Let $d_N$ be the minimum distance of the code $\C_N$. The decoding
scheme can correct $\lfloor\frac{d'}{2}\rfloor$ errors, $d'\geq
d_N-1$, up to an admissible decoding error, if and only if in every
codeword from $c\in\C_N$ of weight $w(c)\leq d'$ all the components
$c_{(N-L)(n-k)+1},\ldots,c_{N(n-k)}$ and $c_{Nn-Lk+1},\ldots,c_{Nn}$
(i.e., those that don't admit a decoding error) are zero.
\end{theorem}

\begin{proof}
Note that the generator matrix $B_N$ is systematic, and the check
matrix of the code is well known to be
\begin{equation*}
H_N=\left(
    \begin{matrix}
        & & & &  D & CB & CAB & \ldots & CA^{N-2}B \\
        & &  &  & 0 & D & CB & \ldots & CA^{N-3}B \\
        & & -Id_{N(n-k)} & & \vdots & \ddots & \ddots & \ddots & \vdots \\
         & & & &  \vdots & & \ddots & \ddots & CB \\
         & & & &  0 & \dots & \dots & 0 & D \\
    \end{matrix}
    \right ).
\end{equation*}
Note that $H_N$ corresponds to the kernel representation matrix
(\ref{eq:rosen}), up to a change of ordering of time indices to
decreasing order and removing the first block row (\ref{eq:rosen}),
that corresponds to the zero return condition.

The minimum distance $d_N$ of the code $\C_{N}$ is precisely the
minimum number of linearly dependent columns of $H_N$, \cite{MS:77},
as the coefficients of one such linear dependency would be the
components of a codeword from $\C_{N}$. This bounds the number of
errors that can be corrected in an $N$-step. Note, however, that
after every $N$-step decoding, the method updates the partial
solution just with $u_t,\ldots,u_{t+L-1}$,  i. e., decoding errors
that occur in the components corresponding to
$u_{t+L},\ldots,u_{t+N}$ (and hence also in those corresponding to
$y_{t+L},\ldots,y_{t+N}$) are admissible. The set of components
corresponding to these vectors that allow errors is
$\alpha=\{1,\ldots,(N-L)(n-k),N(n-k)+1,\ldots,Nn-Lk\}$. Let us
denote its complementary by $\bar{\alpha}$.

An admissible error, a vector with support in $\alpha$, is also a
codeword: an admissible error is the difference between the codeword
sent and the codeword wrongly decoded, and by linearity the
difference of two codewords is also a codeword.

Let us assume that the code doesn't allow to correct error vectors
of weight $t'=\lfloor\frac{d'}{2}\rfloor$ up to an admissible
decoding error. Then, there exists a vector $v$ such that for two
different codewords $c,c'\in\C_{B_N}$ it can be written as $v=c+e$
and $v=c'+e'$ with $w(e),w(e')\leq t'$. Since decoding up to an
admissible error is not possible, we have that $c_{\bar{\alpha}}\neq
c'_{\bar{\alpha}}$, i. e., $e_{\bar{\alpha}}\neq e'_{\bar{\alpha}}$.
Then, $c+e=c'+e'$ and by linearity $c-c'=e'-e=c''$ is a codeword
from $\C_{N}$ with weight $w(c'')\leq w(e)+w(e')\leq 2t'\leq d'$ and
such that $c''_{\bar{\alpha}}\neq 0$, which contradicts the
assumption of the theorem. The inverse implication is immediate.
\end{proof}

As a consequence, the decoding property of our algorithm is as
follows.

\begin{corollary}
Let $L$, $N$, $\C_N$ and $d'$ be as in Theorem
\ref{th:admisible_errors}. Then:
\begin{enumerate}
\item
The output of the algorithm is a codeword from the
convolutional code.
\item
If in every subsequence
\begin{equation*}
    (\tilde{c}_{jL},\tilde{c}_{jL+1},\ldots,\tilde{c}_{jL+N-1}),\quad
    j\geq 0,
\end{equation*}
of the received sequence the Hamming weight of the error is at most
$d'$, then the algorithm recovers the original convolutional
codeword.
\end{enumerate}
\end{corollary}

\begin{proof}
(1)  Every subsequence $(c_{jL},c_{jL+1},\ldots,c_{(j+1)L-1})$ that
is generated by the algorithm at each step is in the kernel of
$H_L$. Moreover, considering the way in which $x_{t+L}$ is updated
in the algorithm, it follows that the sequence $\{c_t\}_{t=0}^T$ is
in the right kernel of the submatrix obtained by removing the first
block row in (\ref{eq:rosen}). The last step of the algorithm
consists in extending the sequence so that the last state becomes
zero.  Therefore the output of the algorithm is also in the kernel
of the first block row of (\ref{eq:rosen}) and is therefore a
codeword.

(2) is a direct consequence of Theorem \ref{th:admisible_errors}.
\end{proof}

\begin{example}
Let us consider the  convolutional code over $\F_5$ generated by the
matrix
\begin{equation*}
\begin{pmatrix}
    1 & 4+z \\
    3 & z \\
    1 & 0
\end{pmatrix} =
\begin{pmatrix}
    P(z) \\
    Q(z)
\end{pmatrix}
\end{equation*}
which as we have seen before can be regarded as the linear system
described by the equations
\begin{equation*}
  \begin{array}{l}
    x_{t+1}= (1,2)\, u_t\\
    y_t = 4 x_t + (1,3)\, u_t\\
    x_0=0
  \end{array}\; ,
\end{equation*}
i. e., it has a minimal realization $(A,B,C,D)=( (0), (1,2), (4),
(1,3))$.

Let us fix $N=2$. Then our algorithm will work with the received
vectors $v_t$, $v_{t-1}$ at each time instant $t$, when a vector
must be decoded with respect to the code that has as check matrix
\begin{equation*}
H_N=\begin{pmatrix}
    1 & 0 & 1 & 3 & 4 & 3\\
    0 & 1 & 0 & 0 & 1 & 3
\end{pmatrix}
\end{equation*}
where the columns 1, 3, 4 correspond to the coordinates of $v_t$. We
observe that although the minimum distance of the the code is 2,
there is no codeword of weight $\leq 2$ with support in the
positions 2, 5, 6. Then if we fix $L=1$ we can allow errors in
coordinates 1, 3, 4 ($v_{t-1}$ will be correctly decoded) and in
exchange be able to correct one error. In this way our scheme will
produce the correct $v_{t-1}$ in each decoding step.

The convolutional code has parameters $[n,k,\delta,\df]=[3,2,1,3]$,
and in particular it allows the correction of one error. Hence our
algorithm takes full advantage of the error correcting capacities of
the code.
\end{example}

\section{Appendix: Uniqueness of decoded sequences}\label{multiplesolutions}

An important aspect of decoding is whether it has one or more
solutions. Since our algorithm works sequentially we are interested
to know whether there is only one list of vectors
$u_t,\ldots,u_{t+L-1}$ in each block decoding step. Note that since
the matrix $B_N$ has maximum rank, then $u_t,\ldots,u_{t+L-1}$ are
uniquely determined by a given block codeword at each step of the
algorithm. Hence the question reduces to the one of how many
convolutional code words are closest to a received message.

\begin{example}\label{ex:double_solution}
Let us consider over $\F_2$ the finite time tracking problem problem
with $N=L=1$ defined by the matrices
\begin{equation*}
    A=
    \begin{pmatrix}
        1 & 0 \\ 0 & 1
    \end{pmatrix}\;
    B=
    \begin{pmatrix}
        0 & 1 \\ 1 & 1
    \end{pmatrix}\;
    C=
    \begin{pmatrix}
        0 & 1 \\ 1 & 1
    \end{pmatrix}\;
    D=
    \begin{pmatrix}
        0 & 1 \\ 1 & 1
    \end{pmatrix}\;
\end{equation*}
Assume that the problem is solved up to a time instant $t$ and let
$x_t=(1,0)$, $\tilde{u}_{t+1}=(0,0)$ and $\tilde{y}_{t+1}=(0,0)$.
We need to decode the vector $z_{t,1}=(0,1,0,0)^\top$ with respect
to the code generated by the matrix
\begin{equation*}
    B_1=
    \begin{pmatrix}
        D \\ I_2
    \end{pmatrix} =\tiny
    \begin{pmatrix}
        0 & 1 \\ 1 & 1 \\ 1 & 0 \\ 0 & 1
    \end{pmatrix}\;.
\end{equation*}
However both the vectors $(0,\underline{0},0,0)^\top$ and
$(0,1,\underline{1},0)^\top$ belong to the code and are at a Hamming
distance of 1 from the vector that we want to decode. Hence both the
corresponding values for $u_t$, $(0,0)$ and $(1,0)$, would be
equally valid.

We consider then a problem with window length $N=2$. Let
$\tilde{u}_{t+2}=(0,0)$, $\tilde{y}_{t+2}=(1,0)$. Then we have to
decode the vector $(1,1,0,1,0,0,0,0)^\top$ with respect to the code
generated by the matrix
\begin{equation*}
    B_2=
    \begin{pmatrix}
        D & CB \\ 0 & D \\ I & 0 \\ 0 & I
    \end{pmatrix}
    = \scriptsize
    \begin{pmatrix}
        0 & 1 & 1 & 1 \\ 1 & 1 & 1 & 0 \\ 0 & 0 & 0 & 1 \\ 0 & 0 & 1 & 1 \\
        1 & 0 & 0 & 0 \\ 0 & 1 & 0 & 0 \\ 0 & 0 & 1 & 0 \\ 0 & 0 & 0 & 1
    \end{pmatrix}\;.
\end{equation*}
The (unique) closest codeword to that vector is
$(1,1,0,1,0,0,\underline{1},0)^\top$ and hence there is a unique
solution which yields $u_t=(1,0)$.
\end{example}

To study the probabilities of multiple solutions we count the number
of vectors that can be uniquely decoded, which are those inside the
largest disjoint balls, centered at the codewords.

Consider a code $\C$ of length $n$, dimension $k$ and minimum
distance $d$ defined over the field $\F$ with $q$ elements. Its
error correction capacity is $t=\lfloor\frac{d-1}{2}\rfloor$. Each
ball of radius $t$ contains exactly $\ds{\sum_{i=0}^t}{n \choose
i}(q-1)^i$ vectors, and there are $q^k$ of such balls (as many as
codewords). Then, the ratio of the number of uniquely decodable
vectors with respect to the cardinality of the whole ambient space
\nolinebreak $\F^n$, known as the \emph{density} of $\C$, is
\begin{equation*}
    \delta_\C=
\frac{q^k\ds{\sum_{i=0}^t}
    {n \choose i}(q-1)^i}{q^{n}}=
    \frac{\ds{\sum_{i=0}^t}
    {n \choose i}(q-1)^i}{q^{n-k}}\;.
\end{equation*}

The probability that a randomly chosen vector is out of all these
balls is $P^o_{\C} = 1-\delta_\C$.

Before addressing the next result, and for the sake of simplicity,
we fix the following notation
\begin{equation*}
    E_{k,t}=\ds{\sum_{i=0}^{t}}{k \choose i}(q-1)^i\;.
\end{equation*}

\begin{theorem}\label{th:prob_multisolution}
Given a finite horizon tracking problem of window length $N$, the
probability that there are $M$ different solutions which differ in
$\Delta$ consecutive vectors can be upper-bounded by
\begin{equation*}
 \frac{\delta^{M-1}_{\mathcal{C}_1}}{E_{k,t}}\ds{\prod_{i=N}^{\Delta}} P^o_{\mathcal{C}_i}.
\end{equation*}
\end{theorem}

\begin{proof}
To have $M$ different solutions that differ in $\Delta$ consecutive
vectors means that at some step of the algorithm all finite horizon
tracking problems with window length $\leq \Delta$ have more than
one solution, and that they cannot be discriminated with a larger
window, i. e., all solutions agree in the last input vector
$u_{t+\Delta}$.\\

On each tracking problem with window lengths $N\leq l \leq \Delta$
the vector to be decoded is further than the error correction
capacity of the code $\C_l$, and as seen before, the probability of
this to happen is $P^o_{\C_l}$. Since this is the case for all
$l=N,\ldots,\Delta$, the probability that all finite horizon
tracking problems with window length $\leq \Delta$ have more than
one solution can be upper-bounded by
\begin{equation}\label{eq:prob_bound}
    P^o_{{\mathcal{C}_N}}{\cdot}\ldots{\cdot}P^o_{{\mathcal{C}_{\Delta}}}=
    \ds{\prod_{i=N}^{\Delta}} P^o_{\mathcal{C}_i}\quad .
\end{equation}

Let us study now the probability that $M$ different optimal
solutions of a finite horizon tracking problem of length $\Delta$
have the same solution vector $u_{t+\Delta}$. Let
$w_{t,N}=(w_{1,\Delta},w_{1,\Delta-1},\ldots,w_{1,1},w_{2,\Delta},w_{2,\Delta-1},\ldots,w_{2,1})$
be the vector to be decoded. For each of the $M$ solutions,
$\{u_{t+i}^j\}_{i=0}^\Delta$ with $j\leq M$ and
$u_{t+\Delta}^1=\ldots=u_{t+\Delta}^M$, let
\begin{equation*}
    w^j_{1,\Delta}=w_{1,\Delta} - \ds{\sum_{i=1}^\Delta} CA^{i-1}B
    u^j_{t+\Delta-i}\quad .
\end{equation*}
Note that the fact that the $M$ solutions have the same
$u_{t+\Delta}$ is equivalent to the $M$ vectors
$v^j=(w^j_{1,\Delta},w_{2,\Delta})$ are decoded to the same codeword
of the block code $\C_1$ and all of them have error of the same
weight (since all $M$ solutions up to instant $\Delta-1$ are also
optimal, otherwise some would have been discarded for smaller window
lengths, and hence contribute the same to the cost functional).

Then, let $c\in\C_{B_1}$ be the codeword to which $v^1$ is decoded,
$c=(c_1,c_2)$ according to the splitting of the $v^j$. We have
$d(c,v^1)=e\leq t=\lfloor\frac{d-1}{2}\rfloor$. The probability that
for all $j=2,\ldots, M$, $v^j$ is also decoded to $c$ and
$d(c,v^j)=e$ depends on:

\begin{itemize}
\item
the probability that there are $\alpha$ error components in the
part $w_{2,\Delta}$ (which is a common part of length $k$ for all
$v^j$)
\begin{equation*}
    P(d(c_2,w_{2,\Delta})=\alpha)=
    \frac{\#\{v|d(c_2,v)=\alpha\}}{\#\{v|d(c_2,v)\leq t\}}=
    \frac{{k\choose \alpha}(q-1)^\alpha}
    {\ds{\sum_{r=0}^t} {k\choose r}(q-1)^r}:=P_{c,\alpha}
\end{equation*}
\item
the probability that $d(c,v^1)=e=\alpha+\beta$ ($e\leq t$) provided
that $d(c_2,w_{2,\Delta})=\nlb\alpha$
\begin{equation*}
    P(d(c,v^1)=\alpha+\beta\,|\,d(c_2,w_{2,\Delta})=\alpha)=
    \frac{{n-k\choose \beta}(q-1)^{\beta}}
    {\ds{\sum_{s=0}^{t-\alpha}} {n-k\choose s}(q-1)^s}:=P_{e|c,\alpha}
\end{equation*}
\item
the probability that for each $j=2,\ldots,M$ $d(c,v^j)=e$ provided
that $d(c_2,w_{2,\Delta})=\alpha$ and $d(c,v^1)=e$
\begin{equation*}
    P(d(c,v^j)=e\,|\,d(c_2,w_{2,\Delta})=\alpha, d(c,v^1)=e)=
    \frac{{n-k\choose \beta}(q-1)^{\beta}}{q^{n-k}}:=P_{v^j|c,\alpha,e}
\end{equation*}
\end{itemize}

Considering all the possibilities for values of $\alpha$ and $\beta$
(and hence of $e$) we have that the probability that all vectors
$v^j$ are decoded to the same codeword of $\C_1$ and that their
errors have the same weight is
\begin{equation}\label{eq:probv'M-1}
  \begin{array}{l}
    \ds{\sum_{\alpha=0}^t\sum_{\beta=0}^{t-\alpha}} P_{c,\alpha} \cdot P_{e|c,\alpha}
    \cdot \ds{\prod_{j=2}^M} P_{v^j|c,\alpha,e}=\\
    \frac{\ds{\sum_{\alpha=0}^t} {k\choose \alpha}(q-1)^\alpha
    \ds{\sum_{\beta=0}^{t-\alpha}}\left({n-k\choose
    \beta}(q-1)^{\beta}\right)^M
    \left(\ds{\sum_{s=0}^{t-\alpha}} {n-k\choose s}(q-1)^s\right)^{-1}}
    {q^{(M-1)(n-k)}\ds{\sum_{r=0}^t} {k\choose r}(q-1)^r} \;\;.
  \end{array}
\end{equation}

Taking into account that
\begin{equation*}
\hspace{-.5cm}
\begin{matrix}
    \ds{\sum_{\beta=0}^{t-\alpha}}\left({n-k\choose \beta}(q-1)^{\beta}\right)^M
    \leq \left(\ds{\sum_{\beta=0}^{t-\alpha}}{n-k\choose \beta}(q-1)^{\beta}\right)^M
    \\{}\\
    {k\choose \alpha}(q-1)^\alpha\leq ({k\choose \alpha}(q-1)^\alpha)^{M-1}
    \\{}\\
    \ds{\sum_{\alpha=0}^t\left({k\choose \alpha}(q-1)^\alpha\right)^{M-1}}
    \left(\ds{\sum_{\beta=0}^{t-\alpha}} {n-k\choose \beta}(q-1)^{\beta}\right)^{M-1}
    \leq \left(\ds{\sum_{\alpha=0}^t {k\choose \alpha}}(q-1)^\alpha \ds{\sum_{\beta=0}^{t-\alpha}}
    {n-k\choose \beta}(q-1)^{\beta}\right)^{M-1}
    \\{}\\
    \ds{\sum_{\alpha=0}^t}{k\choose \alpha}(q-1)^\alpha
    \ds{\sum_{\beta=0}^{t-\alpha}}{n-k\choose \beta}(q-1)^{\beta} \leq
    \ds{\sum_{i=0}^t}{n \choose i}(q-1)^i\quad
    (\text{equality }\Leftrightarrow k,n-k\geq t)
\end{matrix}
\end{equation*}
then (\ref{eq:probv'M-1}) is upper-bounded by
\begin{equation}\label{eq:probM-1bound}
    \frac{\left(\ds{\sum_{\alpha=0}^t}{k\choose \alpha}(q-1)^\alpha
    \ds{\sum_{\beta=0}^{t-\alpha}}{n-k\choose \beta}(q-1)^{\beta}\right)^{M-1}}
    {E_{k,t}\; q^{(M-1)(n-k)}} \leq \frac{\delta_{\mathcal{C}_1}^{M-1}}{E_{k,t}}\;\;.
\end{equation}

Thus, the probability that there are $M$ different solutions which
differ in $\Delta$ consecutive vectors is upper-bounded by the
product of (\ref{eq:prob_bound}) and (\ref{eq:probM-1bound}):
\begin{equation}\label{eq:bound}
 \frac{\delta^{M-1}_{\mathcal{C}_1}}{E_{k,t}}\ds{\prod_{i=N}^{\Delta}} P^o_{\mathcal{C}_i}.
\end{equation}

\end{proof}

\section{Conclusion}

We develop a system theoretic approach towards convolutional
decoding, following the well-known interpretation of convolutional
codes as linear systems. The Bellman optimality principle, applied
to optimizing the Hamming distance function  for linear systems over
finite fields, then yields an optimal control decoding algorithm
that is closely related to the Viterbi algorithm.

To obtain an algorithm with lower computational cost, we propose a
model predictive control algorithm, using a receding horizon
iteration. This new algorithm has good decoding properties, as it
yields desired codeword as long as there are not too many errors on
$N$ consecutive received vectors. We also estimate the probability
that the algorithm computes a unique solution.

\bibliographystyle{amsplain}
\bibliography{../Biblio/biblio,../Biblio/biblioAG,../Biblio/biblioST}
\providecommand{\bysame}{\leavevmode\hbox
to3em{\hrulefill}\thinspace}
\providecommand{\MR}{\relax\ifhmode\unskip\space\fi MR }
\providecommand{\MRhref}[2]{
 \href{http://www.ams.org/mathscinet-getitem?mr=#1}{#2}
}
\providecommand{\href}[2]{#2}

\end{document}